\documentclass{article}[10pt]
\usepackage{amsmath,amsfonts,amsthm,amssymb,dsfont}
\usepackage{fullpage}
\usepackage{soul,color}
\usepackage{graphicx,float,wrapfig}
\usepackage{mathrsfs}
\usepackage[usenames,dvipsnames]{pstricks}
\usepackage{epsfig}
\usepackage{pst-grad} 
\usepackage{pst-plot} 
\usepackage[linesnumbered,boxed,ruled,vlined]{algorithm2e}

\newtheorem{theo}{Theorem}[section]

\newtheorem{lemma}[theo]{Lemma}
\newtheorem{claim}{Claim}

\newtheorem{cor}[theo]{Corollary}
\newtheorem{defi}[theo]{Definition}

\newenvironment{proofof}[1]{\begin{proof}[Proof of #1]}{\end{proof}}

\newcommand{\Ex}{\mathbb{E}}
\everymath{\displaystyle}

\newcommand{\distr}{\mathcal{D}}

\newcommand{\BPPPATH}{\mathsf{BPP_{path}}}
\newcommand{\SZK}{\textsf{SZK}}
\newcommand{\BQP}{\textsf{BQP}}

\newcommand{\AM}{\textsf{AM}}

\newcommand{\PTIME}{\textsf{P}}
\newcommand{\QSZK}{\textsf{QSZK}}
\newcommand{\PH}{\textsf{PH}}

\newcommand{\Fdi}{\mathcal{F}}
\newcommand{\Uni}{\mathcal{U}}

\newcommand{\FdiP}{\mathcal{F}'}
\newcommand{\UniP}{\mathcal{U}'}

\newcommand{\Fore}{\mathsf{For}}
\newcommand{\Fcheck}{\textsf{Forrelation}}

\newcommand{\Simon}{\textsf{Simon}}

\newcommand{\PSimon}{\textsf{Simon}}

\newcommand{\Collision}{\mathsf{Col}}

\newcommand{\PCollision}{\textsf{Collision}}

\newcommand{\Dperm}{\mathcal{P}_{1 \to 1}}
\newcommand{\Dtwoone}{\mathcal{P}_{2 \to 1}}

\newcommand{\Ada}{\mathsf{Ada}}

\newcommand{\poly}{\operatorname*{poly}}
\newcommand{\polylog}{\operatorname*{polylog}}

\newcommand{\Oracle}{\mathcal{O}}
\newcommand{\support}{\mathbf{support}}

\newcommand{\chk}[1]{#1_{\mathsf{chk}}}

\title{A Note on Oracle Separations for $\BQP$}
\date{}
\author{
	Lijie Chen\\
	\small Tsinghua University\\
	\small \texttt{wjmzbmr@gmail.com}
}

\begin{document}
		\maketitle
	\begin{abstract}
	In 2009, using the \textsf{Fourier Checking}\footnote{also called the \Fcheck\  problem in \cite{aaronson2015forrelation}.} problem, Aaronson~\cite{aaronson2010bqp} claimed to construct the relativized worlds such that $\BQP \not\subset \BPPPATH$ and $\BQP \not\subset \SZK$. However, there are subtle errors in the original proof. In this paper, we point out the issues, and rescue these two separations by using more sophisticated constructions.
	
	Meanwhile, we take the opportunity to study the complexity classes $\BPPPATH$ and \SZK. We give general ways to construct functions which are hard for \SZK\ and $\BPPPATH$ (in the query complexity sense). Using these techniques, we give alternative construction for the oracle separation $\BQP \not\subset \SZK$\footnote{It is a folklore result that the quantum walk problem in \cite{childs2003exponential} implies such a separation, see also http://www.scottaaronson.com/blog/?p=114.}, using only Simon's problem~\cite{simon1997power}. We also give new oracle separations for $\PTIME^{\SZK}$ from $\BPPPATH$ and $\PTIME^{\SZK}$ from $\QSZK$. The latter result suggests that $\PTIME^{\SZK}$ might be strictly larger than \SZK. 
	\end{abstract}
	\section{Introduction}
It has been a longstanding open problem in quantum complexity theory to find an oracle separation between $\BQP$ and $\PH$ (see Aaronson~\cite{aaronson2010bqp}). Nevertheless, the current frontier towards this goal is the claimed oracle separation between $\BQP$ and $\BPPPATH$ as shown in \cite{aaronson2010bqp}. In fact, we don't even have an oracle separation between $\BQP$ and $\AM$!

However, we find some subtle errors in the proof for the proposed oracle separations $\BQP^{\Oracle} \not\subset \BPPPATH^{\Oracle}$ and $\BQP^{\Oracle} \not\subset \SZK^{\Oracle}$. It is claimed in \cite{aaronson2010bqp} that almost $k$-wise independence fools $\SZK$ protocols and $\BPPPATH$ machines, yielding the desired oracle separation. Unfortunately, the proof is not correct (see Section~2 for a discussion). But we are also unable to construct a counterexample to either part of this claim. Personally, we feel like almost $k$-wise independence should fool $\SZK$ protocols, but not $\BPPPATH$ machines. Constructing counterexamples or proving either part of the original claim would be an interesting open problem.

In this paper, we rescue these two oracle separations by adding one more twist. For the $\BPPPATH$ case, we show that $\BPPPATH$ machines are unable to distinguish {\em perfectly} two almost $k$-wise independent distributions. Based on that, we prove that the \Fcheck\ problem \cite{aaronson2010bqp,aaronson2015forrelation} is hard for $\BPPPATH$ algorithms (in the query complexity sense), hence the oracle separation $\BQP^{\Oracle} \not\subset \BPPPATH^{\Oracle}$ follows directly. 

For the $\SZK$ case, we find surprisingly that a simple variant of the recent cheat sheet construction by Aaronson, Ben-David and Kothari  \cite{aaronson2015separations} can be used here. Using our new simple construction, we give oracle separation $\BQP^{\Oracle} \not\subset \SZK^{\Oracle}$ using {\em only \PSimon's problem}. Our construction can also be used to prove separation from $\QSZK$, and it works by a black-box fashion: Given any function $f$ with large $R(f)$ ($Q(f)$), we can construct its check-bit version $\chk{f}$, which is hard for any $\SZK$ ($\QSZK$) protocols in the query complexity sense. Utilizing this new tool, we are able to give the new oracle separation $\PTIME^{\SZK^{\Oracle}} \not\subset \QSZK^{\Oracle}$, which is the first non-trivial oracle separation for $\QSZK$ to the best of our knowledge. This also give the oracle evidence that $\SZK$  is strictly contained in $\PTIME^{\SZK}$ as $\SZK \subseteq \QSZK$.

Finally, we establish a method to construct problems which are hard for $\BPPPATH$ algorithms (in the query complexity sense). From this we immediately have the new oracle separation  $\PTIME^{\SZK^{\Oracle}} \not\subset \BPPPATH^{\Oracle}$.
	
	\section{The Issues in the Proof}

In this section we discuss the issues in the proof of~\cite{aaronson2010bqp}.

The proposed separations $\BQP \not\subset \BPPPATH$ and $\BQP \not\subset \SZK$ in \cite{aaronson2010bqp} are based on the following key lemma, which we restate here and recap its original proof in \cite{aaronson2010bqp} for convenience.

\vspace{0.2cm}
\noindent
{\bf Lemma~20~\cite{aaronson2010bqp}}
{\em	
Suppose a probability distribution $\mathcal{D}$ over oracle
strings is $1/t\left(  n\right)  $-almost\ $\operatorname*{poly}\left(
n\right)  $-wise independent,\ for some superpolynomial function $t$. \ Then
no $\mathsf{BPP}_{\mathsf{path}}$\ machine or $\mathsf{SZK}$\ protocol can
distinguish $\mathcal{D}$\ from the uniform distribution $\mathcal{U}$\ with
non-negligible bias.
}

\begin{proof}[Proof in \cite{aaronson2010bqp}]
	Let $M$ be a $\mathsf{BPP}_{\mathsf{path}}$\ machine, and let $p_{\mathcal{D}%
	}$\ be the probability that $M$ accepts an oracle string drawn from
	distribution $\mathcal{D}$. \ Then $p_{\mathcal{D}}$\ can be written as
	$a_{\mathcal{D}}/s_{\mathcal{D}}$, where $s_{\mathcal{D}}$\ is the fraction of
	$M$'s computation paths that are postselected, and $a_{\mathcal{D}}$\ is the
	fraction of $M$'s paths that are both postselected and accepting. \ Since each
	computation path can examine at most $\operatorname*{poly}\left(  n\right)
	$\ bits\ and $\mathcal{D}$ is $1/t\left(  n\right)  $%
	-almost\ $\operatorname*{poly}\left(  n\right)  $-wise independent, we have%
	\[
	1-\frac{1}{t\left(  n\right)  }\leq\frac{a_{\mathcal{D}}}{a_{\mathcal{U}}}%
	\leq1+\frac{1}{t\left(  n\right)  }~\ \ \text{and \ }1-\frac{1}{t\left(
		n\right)  }\leq\frac{s_{\mathcal{D}}}{s_{\mathcal{U}}}\leq1+\frac{1}{t\left(
		n\right)  }.
	\]
	Hence%
	\[
	\left(  1-\frac{1}{t\left(  n\right)  }\right)  ^{2}\leq\frac{a_{\mathcal{D}%
		}/s_{\mathcal{D}}}{a_{\mathcal{U}}/s_{\mathcal{U}}}\leq\left(  1+\frac
	{1}{t\left(  n\right)  }\right)  ^{2}.
	\]

	Now let $P$ be an $\mathsf{SZK}$\ protocol. \ Then by a result of Sahai and
	Vadhan \cite{sahai2003complete}, there exist polynomial-time samplable distributions $A$\ and
	$A^{\prime}$\ such that if $P$ accepts, then $\left\Vert A-A^{\prime
	}\right\Vert \leq1/3$, while if $P$ rejects, then $\left\Vert A-A^{\prime
}\right\Vert \geq2/3$. \ But since each computation path can examine at most
$\operatorname*{poly}\left(  n\right)  $\ oracle bits\ and $\mathcal{D}$ is
$1/t\left(  n\right)  $-almost\ $\operatorname*{poly}\left(  n\right)  $-wise
independent, we have $\left\Vert A_{\mathcal{D}}-A_{\mathcal{U}}\right\Vert
\leq1/t\left(  n\right)  $\ and $\left\Vert A_{\mathcal{D}}^{\prime
}-A_{\mathcal{U}}^{\prime}\right\Vert \leq1/t\left(  n\right)  $, where the
subscript denotes the distribution from which the oracle string was drawn.
\ Hence%
\[
\left\vert \left\Vert A_{\mathcal{D}}-A_{\mathcal{D}}^{\prime}\right\Vert
-\left\Vert A_{\mathcal{U}}-A_{\mathcal{U}}^{\prime}\right\Vert \right\vert
\leq\left\Vert A_{\mathcal{D}}-A_{\mathcal{U}}\right\Vert +\left\Vert
A_{\mathcal{D}}^{\prime}-A_{\mathcal{U}}^{\prime}\right\Vert \leq\frac
{2}{t\left(  n\right)  }%
\]
and no $\mathsf{SZK}$\ protocol\ exists.
\end{proof}

Now we discuss the subtle errors in the proof above. For the $\BPPPATH$ case, the problem is that  $p_{\distr}$ {\em can't} be written as $a_{\distr}/s_{\distr}$. On an input $x$, let $s_x$ denote the fractions of $M$'s computation paths that are postselected, and $a_x$ denote the fractions of $M$'s computation paths that are postselected and accepting, then the probability that $M$ accepts $x$ is $a_x/s_x$. So the probability $p_{\distr}$ that $M$ accepts an oracle string drawn from $\distr$ is in fact $\Ex_{x \sim \distr} [a_x/s_x]$, which certainly does not equal $a_\distr/s_\distr$.

For the $\SZK$ case, the problem is that the statement
\[
\left\vert \left\Vert A_{\mathcal{D}}-A_{\mathcal{D}}^{\prime}\right\Vert
-\left\Vert A_{\mathcal{U}}-A_{\mathcal{U}}^{\prime}\right\Vert \right\vert
\leq\frac
{2}{t\left(  n\right)  }%
\]
does not mean there are no $\SZK$ protocols to distinguish $\distr$ and $\mathcal{U}$. 
Let $A(x)$ and $A'(x)$ be the distributions with input $x$. By the definition, $A_{\distr} = \Ex_{x\sim\distr}[A(x)]$ and $A'_{\distr} = \Ex_{x \sim \distr}[A'(x)]$. Let $p_{\distr}$ be the probability that protocol $P$ accepts an input drawn from $\distr$, which is by definition, $p_{\distr} = \Pr_{x\sim\distr}[\|A(x)-A'(x)\| \le 1/3]$. The intended argument seems like:  
$
\left\vert \left\Vert A_{\mathcal{D}}-A_{\mathcal{D}}^{\prime}\right\Vert
-\left\Vert A_{\mathcal{U}}-A_{\mathcal{U}}^{\prime}\right\Vert \right\vert
$ is small implies $|p_{\distr} - p_{\mathcal{U}}|$ is small too.

But this claim is not correct. Consider the following toy example, let the input domain be $\{0,1,2,3\}$, output domain be $\{0,1,2\}$, and define distributions $A(x)$ and $A'(x)$ as follows:
$$
A(x) = \begin{cases}
\text{constant distribution on } $\{x\}$ & \quad \text{when }x \in \{0,1,2\}\\

\text{uniform distribution over } $\{0,1,2\}$& \quad \text{when }x =3
\end{cases}
$$
and
$$
A'(x) = \begin{cases}
\text{constant distribution on } \{(x+1)\bmod 3\}  & \quad \text{when }x \in \{0,1,2\}\\

\text{uniform distribution over } $\{0,1,2\}$& \quad \text{when }x =3
\end{cases}
$$

Now let $\mathcal{U}$ be the uniform distribution over $\{0,1,2\}$, and $\mathcal{D}$ be the constant distribution on $\{3\}$. We can see that both $\|A_{\mathcal{U}} - A'_{\mathcal{U}}\| $ and  $\|A_{\mathcal{D}} - A'_{\mathcal{D}}\| $ are zero, while clearly $p_{\mathcal{U}} = 0$ and $p_{\mathcal{D}} = 1$.



	\section{Preliminaries}

\subsection{Oracle Separations and Query Complexity}

When proving oracle separations, the standard way is to prove some analogous result in {\em query complexity}, and {\em lift} it to an oracle separation. 

It usually proceeds as follows: Let $N$ denote the input length. We find a problem $\mathcal{P}$ such that any algorithm in complexity class $\mathcal{C}$ needs superlogarithmically many {\em queries} to the input in order to solve it, while there exists an algorithm in complexity class $\mathcal{D}$, which only needs $\polylog(N)$ time. Then by the standard diagonalization method, we can construct an oracle $\Oracle$ such that $\mathcal{D}^{\Oracle} \not\subset \mathcal{C}^{\Oracle}$ {\em unconditionally}.

For convenience, we use $n$ to denote a parameter of the problem size and let $N=N(n)=2^n$.

We use $Q(f)$ to denote the bounded-error quantum query complexity, where the algorithm is only required to be correct with probability at least $2/3$; see the survey on query complexity by Buhrman and de Wolf \cite{buhrman2002complexity} for the formal definition.

\subsection{Complexity Classes}

We assume familiarity with some standard complexity classes like \BQP, \SZK, \QSZK\ and \AM. For completeness, we introduce the less well-known complexity class $\BPPPATH$.


Roughly speaking, $\BPPPATH$ consists of the computational problems can be solved in probabilistically polynomial time, given the ability to {\em postselect} on some event (which may happen with a very small probability). Formally:

\begin{defi}
	$\mathsf{BPP}_{\mathsf{path}}$ (defined by Han, Hemaspaandra, and Thierauf \cite{han1997threshold}) is the class of languages $L\subseteq\left\{
	0,1\right\}  ^{\ast}$\ for which there exists a $\mathsf{BPP}$\ machine $M$,
	which can either \textquotedblleft succeed\textquotedblright\ or
	\textquotedblleft fail\textquotedblright\ and conditioned on succeeding either
	\textquotedblleft accept\textquotedblright\ or \textquotedblleft
	reject,\textquotedblright\ such that for all inputs $x$:
	
	\begin{enumerate}
		\item[(i)] $\Pr\left[  M\left(  x\right)  \text{ succeeds}\right]  >0$.
		
		\item[(ii)] $x\in L\Longrightarrow\Pr\left[  M\left(  x\right)  \text{ accepts
		}|~M\left(  x\right)  \text{ succeeds }\right]  \geq\frac{2}{3}$.
		
		\item[(iii)] $x\notin L\Longrightarrow\Pr\left[  M\left(  x\right)  \text{
			accepts }|~M\left(  x\right)  \text{ succeeds }\right]  \leq\frac{1}{3}$.
	\end{enumerate}
\end{defi}
\subsection{Almost $k$-wise Independence and Its Generalizations}

We introduce the concept of almost $k$-wise independence defined in \cite{aaronson2010bqp}, which will be used frequently throughout this paper. We slightly change the old definition so that it applies to distributions over $\{0,1\}^M$ rather than $\{-1,1\}^M$.

Let $Z=z_{1}\ldots z_{M}\in\left\{  0,1\right\}  ^{M}$ be a string. \ Then a
\textit{literal} is a term of the form $z_i$ or $1-z_i$, and a
$k$-\textit{term }is a product of $k$ literals (each involving a different
$z_{i}$), which is $1$ if the literals all take on prescribed values and $0$ otherwise. 
Let $\mathcal{U}$\ be the uniform distribution over $\left\{  0,1\right\}
^{M}$.

\begin{defi} A distribution $\mathcal{D}$ over $\left\{  0,1\right\}  ^{M}$ is
	$\varepsilon$\textit{-almost }$k$\textit{-wise independent} if for every
	$k$-term $C$,%
	\[
	1-\varepsilon\leq\frac{\Pr_{\mathcal{D}}\left[  C\right]  }{\Pr_{\mathcal{U}%
		}\left[  C\right]  }\leq1+\varepsilon.
	\]
	(Note that $\Pr_{\mathcal{U}}\left[  C\right]  $\ is just $2^{-k}$.)
\end{defi}


We also generalize the above concept in the following way.

\begin{defi}
	Given two distributions $\distr_1$ and $\distr_2$ over $\{0,1\}^{M}$, we say $\distr_1$ $\varepsilon$-almost $k$-wise dominates $\distr_2$ if for every $k$-term $C$,
	
	$$
	\frac{\Pr_{\distr_1}\left[  C\right]  }{\Pr_{\distr_2}\left[  C\right]  } \ge 1 - \varepsilon.
	$$
	
	And we say $\distr_1$ and $\distr_2$ are $\varepsilon$-almost $k$-wise equivalent if they $\varepsilon$-almost $k$-wise dominate each other, i.e., for every $k$-term $C$,
	
	$$
 	1 - \varepsilon \leq
	\frac{\Pr_{\distr_1}\left[  C\right]  }{\Pr_{\distr_2}\left[  C\right]  } \leq 1 + \varepsilon.
	$$
\end{defi}

So a distribution $\distr$ is 
$\varepsilon$\textit{-almost }$k$\textit{-wise independent}, iff it is $\varepsilon$-almost $k$-wise equivalent to the uniform distribution $\Uni$. 

\subsection{Problems}
In this subsection we introduce several problems that will be used throughout this paper.

\subsubsection{\Fcheck}
The first one is \Fcheck, which is first defined by Aaronson \cite{aaronson2010bqp}, and studied again by Aaronson and Ambainis \cite{aaronson2015forrelation}. For convenience, we will assume the inputs are in $\{-1,1\}^{M}$ rather than $\{0,1\}^{M}$. 

\begin{defi}[\Fcheck\ problem]
We are given access to two Boolean functions
$f,g:\left\{  0,1\right\}  ^{n}\rightarrow\left\{ -1,1\right\}  $.  We want
to estimate the amount of correlation between $f$\ and the Fourier transform
of $g$---that is, the quantity%
\[
\Phi_{f,g}:=\frac{1}{2^{3n/2}}\sum_{x,y\in\left\{  0,1\right\}  ^{n}}f\left(
x\right)  \left(  -1\right)  ^{x\cdot y}g\left(  y\right)  .
\]
It is not hard to see that $\left\vert \Phi_{f,g}\right\vert \leq1$\ for all
$f,g$. \ The problem is to decide whether $\left\vert \Phi
_{f,g}\right\vert \leq 0.01$\ or $\Phi_{f,g}\geq 0.07$, promised
that one of these is the case. 

We will use $\Fore_n$ to denote the partial function representing the \Fcheck\ problem with parameter $n$ (evaluates to $1$ when $\Phi_{f,g}\geq 0.07$, and $0$ when $|\Phi_{f,g}|\leq 0.01$), whose input length is $2 \cdot 2^n = 2N$.  When $n$ is clear from the context, we use $\Fore$ for simplicity.
\end{defi}


\subsubsection{\PCollision}

We now recall the \PCollision\ problem, which is to decide whether the input is a permutation or is $2$-to-$1$, promised that one of them is the case.

\begin{defi}[\PCollision\ problem]
We are given access to a function $f: [N] \to [N]$, and want to decide whether $f$ is a permutation or a $2$-to-$1$ function, promised that one of these is the case.

Since we are interested in boolean inputs, we can encode its input as a string in $\{0,1\}^{n \cdot N}$ (recall $N = 2^n$), and we use $\Collision_n$ to denote the \PCollision\ problem with parameter $n$ (evaluates to $1$ when the function is $2$-to-$1$, and $0$ when the function is bijective), whose input length is $n \cdot 2^n = N \log N$. When there is no confusion, we use $\Collision$ for simplicity.
\end{defi}

This problem admits a simple $\SZK$ protocol in which the verifier makes only $\poly(n)$ queries to the input. 

In 2002, Aaronson \cite{aaronson2002quantum} proved the first non-constant lower
bound for the \PCollision\ problem: namely, any bounded-error quantum algorithm to solve
it needs $\Omega(N^{1/5})  $\ queries to $f$. Aaronson and
Shi \cite{aaronson2004quantum} subsequently improved the lower bound to $\Omega(
N^{1/3})$, for functions $f:\left[  N\right]  \rightarrow\left[
3N/2\right]  $; then Ambainis \cite{ambainis2005polynomial} and Kutin \cite{kutin2005quantum}
proved the optimal $\Omega(  N^{1/3})  $\ lower bound for functions
$f:\left[  N\right]  \rightarrow\left[  N\right]  $.

\subsubsection{\PSimon's problem}
Finally we recall the definition of the famous \PSimon's problem.

\begin{defi}[\PSimon's problem]
We are given access to a function $f : \{0,1\}^n \to \{0,1\}^n$ and promised that there exists a ``secret string"  $s\in \{0,1\}^n$ such that $y,z\in\{0,1\}^n$, $f(y)=f(z)$ if and only if $y=z$ or $y \oplus z =s$. The problem then is to find s. We can encode its input as a string in $\{0,1\}^{n\cdot N}$, and we use $\Simon_n$ to denote the \PSimon's problem with parameter $n$. When there is no confusion, we use $\Simon$ for simplicity.
\end{defi}

As shown by Simon \cite{simon1997power}, we have a $\poly(n)$ query quantum algorithm computing $\Simon_n$. Furthermore, it is hard for any classical algorithms to compute it even with a small success probability. We will use the following lemma which follows from a classical result.

\begin{lemma}\label{lm:hard-simon}
	Any $N^{o(1)}$-query randomized algorithm can compute $\Simon_n$ with success probability at most $1/\sqrt[3]{N}$.
\end{lemma}

\begin{proof}
	A classical result (see \cite{cleve1999introduction}) shows that any randomized algorithm solving $\Simon_n$ with error probability at most $\epsilon$, needs $\Omega(\sqrt{2^n}\log(1/\epsilon))$ queries. Plugging in $\epsilon = 1 - N^{1/3}$, it follows that any randomized algorithm with success probability at least $1/\sqrt[3]{N}$, need at least $\Omega\left(\sqrt{N} \cdot \log\left(\frac{1}{1-1/\sqrt[3]{N}}\right)\right) = \Omega(N^{1/6})$ queries, and the lemma follows directly.
\end{proof}
	
	\section{Several Input Distributions}

In this section we construct several useful input distributions for the \Fcheck\ problem and the \PCollision\ problem. These will be the main ingredients in our proofs.

\subsection{\Fcheck}

Let $M = 2 \cdot 2^n = 2 \cdot N$. 

We first introduce the forrelated distribution $\Fdi$ on $\{-1,1\}^M$ defined in \cite{aaronson2010bqp}.


\begin{defi}
	A sample
	$\left\langle f,g\right\rangle $ from $\Fdi$ is generated as follows. First choose a random real vector $v=\left(
	v_{x}\right)  _{x\in\left\{  0,1\right\}  ^{n}}\in\mathbb{R}^{N}$, by drawing
	each entry independently from a Gaussian distribution with mean $0$ and
	variance $1$. \ Then set $f\left(  x\right)  :=\operatorname*{sgn}\left(
	v_{x}\right)  $\ and $g\left(  x\right)  :=\operatorname*{sgn}\left(
	\widehat{v}_{x}\right)  $\ for all $x$. \ Here%
	\[
	\operatorname*{sgn}\left(  \alpha\right)  :=\left\{
	\begin{array}
	[c]{cc}%
	1 & \text{if }\alpha\geq0\\
	-1 & \text{if }\alpha<0
	\end{array}
	\right.
	\]
	and $\widehat{v}$ is the Fourier transform of $v$\ over $\mathbb{Z}_{2}^{n}$:%
	\[
	\widehat{v}_{y}:=\frac{1}{\sqrt{N}}\sum_{x\in\left\{  0,1\right\}  ^{n}%
	}\left(  -1\right)  ^{x\cdot y}v_{x}.
	\]
	In other words, $f$\ and $g$\ \textit{individually}\ are still uniformly
	random, but they are no longer independent: now $g$ is extremely well
	correlated with the Fourier transform of $f$ (hence \textquotedblleft
	forrelated\textquotedblright).
\end{defi}

By the simple transformation $x \to \frac{1+x}{2}$, $\Fdi$ can be viewed as distribution on $\{0,1\}^M$. We introduce the following key theorem from \cite{aaronson2010bqp}.

\begin{theo}[Theorem~19 in \cite{aaronson2010bqp}]
	For all $k\leq\sqrt[4]{N}$, the forrelated distribution
	$\mathcal{F}$ is $O\left(  k^{2}/\sqrt{N}\right)  $-almost $k$-wise independent.
\end{theo}

Intuitively, w.h.p., a sample from $\Fdi$ is a 1-input of function $\Fore$ and a sample from $\Uni$ is a $0$-input of $\Fore$. However, the supports of $\Fdi$ and $\Uni$ are not disjoint, which causes some trouble. To fix this problem, we define the following two distributions on $\{-1,1\}^M$.

\begin{defi}
	$\FdiP$ is the conditional distribution obtained by $\Fdi$ conditioned on the event that $\Phi_{f,g} \ge 0.07$, i.e., a sample $\langle f,g \rangle$ from $\FdiP$ can be generated as follows: We draw a sample $\langle f,g \rangle$ from $\Fdi$, if $\Phi_{f,g} \ge 0.07$, we simply output $\langle f,g \rangle$, otherwise we discard $\langle f,g \rangle$ and start again until the requirement is satisfied. 
	
	In the same way, $\UniP$ is the conditional distribution obtained by $\Uni$ conditioned on the event that $|\Phi_{f,g} | \le 0.01$.
\end{defi}

By definition, we can see $\FdiP$ is supported on $1$-inputs to $\Fore$, and $\UniP$ is supported on the $0$-inputs. Furthermore, they are both almost $k$-wise independent.

\begin{lemma}\label{lm:alm-k-indp}
	For any $k = N^{o(1)}$, $\FdiP$ and $\UniP$ are $o(1)$-almost $k$-wise independent.
\end{lemma}

In order to prove the above lemma, we need two concentration results about $\Phi_{f,g}$ on the two distributions $\Fdi$ and $\Uni$.

	\begin{lemma}[Part of Theorem~9 in \cite{aaronson2010bqp}]\label{lm:con-for}
		With $1- 1/\exp(N)$ probability over $\langle f,g \rangle$ drawn from $\Fdi$, we have $\Phi_{f,g} \ge 0.07$.
	\end{lemma} 
	
	\begin{lemma}[Lemma~34 in \cite{aaronson2015forrelation}, simplified]\label{lm:con-uni}
		Suppose $f,g: \{0,1\}^n \rightarrow\left\{  -1,1\right\}  $\ are chosen uniformly at random.
		\ Then%
		\[
		\Pr_{f,g}\left[  \left\vert \Phi_{f,g%
		}\right\vert \geq\frac{t}{\sqrt{N}}\right]  =O\left(  \frac{1}{t^{t}}\right)
		.
		\]
	\end{lemma}

\begin{proof}
	
	By Lemma~\ref{lm:con-for} and the definition of distribution $\FdiP$, for any $k=N^{o(1)}$-term $C$, we have $$\left|\Pr_{\FdiP}[C] - \Pr_{\Fdi}[C]\right| \le 1/\exp(N).$$ Therefore
	
	$$
	\frac{\Pr_{\FdiP}[C]}{\Pr_{ \Uni}[C]} \ge 
	\frac{\Pr_{ \Fdi}[C]}{\Pr_{ \Uni}[C]} - \frac{1/\exp(N)}{\Pr_{ \Uni}[C]} \ge 1 - o(1),
	$$
	
	and similarly,
	$$
	\frac{\Pr_{ \FdiP}[C]}{\Pr_{ \Uni}[C]} \le 
	\frac{\Pr_{ \Fdi}[C]}{\Pr_{ \Uni}[C]} + \frac{1/\exp(N)}{\Pr_{ \Uni}[C]} \le 1 + o(1). 
	$$
	
	Letting $t= \sqrt{N}/100$ and applying Lemma~\ref{lm:con-uni}, we have $\Pr_{f,g}\left[  \left\vert \Phi_{f,g%
	}\right\vert \geq 0.01\right] \le 1/\exp(\sqrt{N})$. So with $1-1/\exp(\sqrt{N})$ probability over $\langle f,g \rangle$ drawn from $\Uni$, we have $|\Psi_{f,g}| \le 0.01$. Then by the definition of $\UniP$, for any $k=N^{o(1)}$-term $C$, we have $\left| \Pr_{ \UniP}[C] - \Pr_{ \Uni}[C] \right| \le 1/\exp(\sqrt{N})$. The claim now follows in the same way as for $\FdiP$.
	
\end{proof}

\subsection{\PCollision}
We can think of the input $f: [N] \to [N]$ as a string $X = x_1,x_2,\dotsc,x_{N}$ in $[N]^N$. Since we are interested in Boolean inputs, we can easily encode such an $X$ as an $m$-bit string where $m=n\cdot N$. Slightly abusing notation, we will speak interchangeably about $X$ as an element in $\{0,1\}^m$ or $[N]^N$.

\begin{defi}
	Let $\Dperm^n$ be the uniform distribution over all permutations on $[N]$, and $\Dtwoone$ be the uniform distribution over all $2$-to-$1$ functions from $[N] \to [N]$. Both can be easily interpreted as distributions over $\{0,1\}^m$. We will also use $\Dperm^n$ and $\Dtwoone^n$ to denote the corresponding distributions over $\{0,1\}^m$ for convenience.
\end{defi}

We have the following important lemma.

\begin{lemma}\label{lm:col-distr}
	For any $k = N^{o(1)}$, $\Dtwoone$ $o(1)$-almost $k$-wise dominates $\Dperm$.
\end{lemma}

In order to prove the above lemma, we need the following technical claim, which shows that almost $k$-wise dominance behaves well with respect to restrictions. Given a $k$-term $C$, let $V(C)$ be the set of variables occurring in $C$. In addition, given a set $S$ of variables such that $V(C) \subseteq S$, let $U_S(C)$ be the set of all $2^{|S|-k}$ terms $B$ such that $V(B) = S$ and $B \implies C$.

\begin{claim}\label{clm:fool-cond}
	Given a $k$-term $C$ and a set $S$ containing $V(C)$, suppose that for every term $B \in U_S(C)$ we have
	$$
	\Pr_{\distr}[B]/\Pr_{\mathcal{U}}[B] \ge 1-\epsilon.
	$$
	Then
	$$
	\Pr_{\distr}[C]/\Pr_{\mathcal{U}}[C] \ge 1-\epsilon.
	$$
\end{claim}

\begin{proof}
	It is easy to see that for any distribution $\distr'$, $\Pr_{\distr'}[C] = \sum_{B \in U_S(C)} \Pr_{\distr'}[B]$, and the claim follows directly.
\end{proof}

Given an input $X = x_1,x_2,\dotsc,x_N$, let $\Delta(x_i,y)$ denote the $n$-term that evaluates to $1$ if and only if $x_i = y$. We say a term $C$ is a proper $k$-term, if it is a product of the form $\Delta(x_{i_1},y_1) \cdots \Delta(x_{i_k},y_k)$, where $1 \le i_1 < \cdots i_k \le N$ and $y_1,\dotsc,y_k \in [N]$. 

We now prove Lemma~\ref{lm:col-distr}.

\begin{proofof}{Lemma~\ref{lm:col-distr}}
	 Note that a Boolean $k$-term can involve bits occurring in at most $k$ different $x_i$'s. So by Claim~\ref{clm:fool-cond}, to show that any Boolean $k$-term $C'$ satisfies $\Pr_{\Dtwoone^n}[C'] \ge (1-o(1)) \cdot \Pr_{\Dperm^n}[C']$, it suffices to show that any {\em proper} $k$-term
	 $$
	 C = \Delta(x_{i_1},y_1) \cdots \Delta(x_{i_k},y_k)
	 $$
	 satisfies $\Pr_{\Dtwoone^n}[C] \ge (1-o(1)) \cdot \Pr_{\Dperm^n}[C]$.
	 
	 If there exist two distinct $a \ne b$ such that $y_a = y_b$, then we immediately have $\Pr_{\Dperm^n}[C] = 0$, and the statement becomes trivial. So we can assume that all $y_i$'s are distinct.
	 
	 Then by a direct calculation, we have
	 
	 $$
	 \Pr_{\Dperm^n}[C] = \prod_{i=0}^{k-1} \frac{1}{N-i}.
	 $$
	 
	 By first fixing the function's image, it is easy to show there are 
	 $$
	 \binom{N}{N/2} \cdot \frac{N!}{2^{N/2}}
	 $$
	 $2$-to-$1$ functions in total.
	 
	 Then we compute how many $2$-to-$1$ functions are compatible with the term $C$. We first fix the image; there are $\binom{N-k}{N/2-k}$ possibilities. Then we pick $k$ other $x_i$'s such that they take values in $\{y_1,y_2,\dotsc,y_k\}$ and assign values to them; there are $\binom{N-k}{k} \cdot k!$ possibilities. Finally, we assign value to the other $N-2k$ $x_i$'s; there are $\frac{(N-2k)!}{2^{N/2-k}}$ possibilities. Putting everything together, there are 
	 
	 $$
	  \binom{N-k}{N/2-k} \cdot \binom{N-k}{k}\cdot k!\cdot \frac{(N-2k)!}{2^{N/2-k}}
	 $$
	 $2$-to-$1$ functions compatible with $C$. Since $\Dtwoone$ is the uniform distribution over all $2$-to-$1$ functions, we have
	 \begin{align*}
	 \Pr_{\Dtwoone^n}[C] &= \left\{ \binom{N-k}{N/2-k} \cdot \binom{N-k}{k}\cdot k!\cdot \frac{(N-2k)!}{2^{N/2-k}} \right\} \Big/ \left\{ \binom{N}{N/2} \cdot \frac{N!}{2^{N/2}} \right\}\\
	  &= \left\{ \frac{(N-k)!}{(N/2)!\cdot(N/2-k)!} \cdot \frac{(N-k)!}{k!\cdot(N-2k)!}\cdot k!\cdot \frac{(N-2k)!}{2^{N/2-k}} \right\} \Big/ \left\{ \frac{N!}{(N/2)!\cdot(N/2)!} \cdot \frac{N!}{2^{N/2}} \right\}\\
	  &= \left\{ \frac{(N-k)!\cdot(N-k)!}{(N/2)!\cdot (N/2-k)!\cdot 2^{N/2-k}} \right\} \Big/ \left\{ \frac{N!\cdot N!}{(N/2)!\cdot(N/2)!\cdot 2^{N/2}}\right\}\\
	  &= \frac{(N-k)!\cdot(N-k)!\cdot (N/2)!\cdot 2^k}{N!\cdot N!\cdot (N/2-k)!}\\
	 &= \prod_{i=0}^{k-1} \frac{1}{N-i} \cdot \prod_{i=0}^{k-1} \frac{N-2\cdot i}{N-i}\\ 
	 &\ge \Pr_{\Dperm^n}[C] \cdot (1 - \sum_{i=0}^{k-1} \frac{i}{N-i})
	 \ge \Pr_{\Dperm^n}[C] \cdot (1 - o(1)). \tag{$k \le N^{o(1)}$}\\
	 \end{align*}
	 
	 This completes the proof.
	 
\end{proofof}
	
	\section{Oracle Separations from $\BPPPATH$}




We first establish a sufficient condition for showing a function is hard for $\BPPPATH$ algorithms.

\begin{theo}\label{theo:hard-BPPPATH}
	Fix a partial function $f : D \to \{0,1\}$ with $D \subset \{0,1\}^M$. Suppose there are two distributions $\distr_0$ and $\distr_1$ supported on $0$-inputs and $1$-inputs respectively, such that they are $o(1)$-almost $k$-wise equivalent. Then there are no $\BPPPATH$ algorithms can compute $f$ using at most $k$ queries.
\end{theo}

\begin{proof}
	Let $M$ be a $\BPPPATH$ machine which computes $f$. Then let $a(x)$  and $s(x)$ be the accepting (success) probability of $M$ on input $x$. For a distribution $\distr$ on $\{0,1\}^M$, let $a(\distr) = \Ex_{x \sim \distr}[a(x)]$ and $s(\distr) = \Ex_{x \sim \distr}[s(x)]$.
	
	By the definition of $\BPPPATH$ algorithms and the fact that $\distr_0$ ($\distr_1$) is supported on $0$-inputs ($1$-inputs), we have $a(\distr_1) \ge 2/3 \cdot s(\distr_1) $ and $a(\distr_0) \le 1/3 \cdot s(\distr_0)$.
		
	Since $M$ makes at most $k$ queries, $a(x)$ can be written as $a(x) = \sum_{i=1}^{m} a_i \cdot C_i(x)$, such that each $C_i$ is a $k'$-term for $k' \le k$ and each $a_i \ge 0$. Therefore, using the fact that $\distr_0$ and $\distr_1$ are $o(1)$-almost $k$-wise equivalent, we have
	$$
	a(\distr_1) = \sum_{i=1}^{m} \Ex_{x \sim \distr_1}[a_i \cdot C_i(x)] \ge (1-o(1)) \cdot \sum_{i=1}^{m}  \Ex_{x \sim \distr_0}[a_i \cdot C_i(x)] = (1-o(1)) \cdot a(\distr_0).
	$$ 
	
	Similarly, we have $a(\distr_0) \ge (1-o(1)) \cdot a(\distr_1)$, hence $ 1-o(1)\le a(\distr_0) / a(\distr_1) \le 1+o(1)$. The same goes for $s(\distr_0)$ and $s(\distr_1)$, so $1-o(1)\le s(\distr_0) / s(\distr_1) \le 1+o(1)$. But this means that $1-o(1)\le \frac{a(\distr_1)/s(\distr_1)}{a(\distr_0)/s(\distr_0)} \le 1+o(1)$, which contradicts the fact that $a(\distr_1)/s(\distr_1) \ge 2/3$ and $a(\distr_0)/s(\distr_0) \le 1/3$. This completes the proof.
\end{proof}

\subsection{$\mathsf{BQP}^{\Oracle} \not\subset \BPPPATH^{\Oracle}$ based on \Fcheck}

Using Theorem~\ref{theo:hard-BPPPATH} and the input distributions $\FdiP$ and $\UniP$ to $\Fore_n$, we can show the function $\Fore_n$ is hard for $\BPPPATH$ algorithms.

\begin{theo}
	There are no $\BPPPATH$ algorithms can compute $\Fore_n$ using $\poly(n)$ queries.
\end{theo}

\begin{proof}
	
	By the definition of $\FdiP$ and $\UniP$, and Lemma~\ref{lm:alm-k-indp}, we can see $\FdiP$ ($\UniP$) is supported on $1$-inputs ($0$-inputs) of $\Fore_n$, and they are both $o(1)$-almost $N^{o(1)}$-wise independent. Which means they are also $o(1)$-almost $N^{o(1)}$-wise equivalent.
	
	Since $\poly(n) = N^{o(1)}$, the theorem follows directly from Theorem~\ref{theo:hard-BPPPATH}.
\end{proof}

Using a standard diagonalization procedure (which we omit the details here), the oracle separation we want follows easily.

\begin{cor}
	There exists an oracle $\Oracle$ such that $\mathsf{BQP}^{\Oracle} \not\subset \BPPPATH^{\Oracle}$.
\end{cor}
\subsection{The Adaptive Construction}

In this subsection we introduce a construction which turns a Boolean function into its adaptive version. 

\begin{defi}[Adaptive Construction]
	Given a function $f :¡¡D \to \{0,1\}$, such that $D \subset \{0,1\}^M$ and an integer $d$, we define $\Ada_{f,d}$, its depth $d$ adaptive version, as follows:
	
	$$
	\Ada_{f,0} := f,
	$$
	and
	
	$$
	\Ada_{f,d}: D \times D_{d-1} \times D_{d-1} \to \{0,1\}
	$$
	$$
	\Ada_{f,d}(w,x,y) := \begin{cases}
	\Ada_{f,d-1}(x) & \quad \text{if } f(w) = 0\\
	\Ada_{f,d-1}(y) & \quad \text{if } f(w) = 1\\
	\end{cases}
	$$
	where $D_{d-1}$ denotes the domain of $\Ada_{f,d-1}$. 
	
	The input to $\Ada_{f,d}$ can be encoded as a string of length $(2^{d+1}-1)\cdot M$. Thus, $\Ada_{f,d}$ is a partial function from $D^{(2^{d+1}-1)} \to \{0,1\}$,
\end{defi}

We can also interpret $\Ada_{f,d}$ more intuitively as follows:  
Given a full binary tree with height $d$, each node encodes a valid input to $f$. The answer is determined by the following procedure: Starting with the root, we compute $f$ with the corresponding input; if it is $0$, we then go to the left child, otherwise we go to the right child. Once we reach a leaf, we output the answer to the input on it.

We have the following theorem, showing that certain functions' adaptive version are hard for $\BPPPATH$ algorithms.

\begin{theo} \label{theo:hard-adapt}
	Fix a function $f : D \to \{0,1\}$ such that $D \subset \{0,1\}^M$. Suppose there are two distributions $\distr_0$ and $\distr_1$ supported on $0$-inputs and $1$-inputs respectively, such that $\distr_1$ $o(1)$-almost $k$-wise dominates $\distr_0$ for any $k \le \poly(n)$. Then no $\poly(n)$-time $\BPPPATH$ algorithms can compute $\Ada_{f,n}$.
\end{theo}

\begin{proof}
	
	We first discuss some properties for a function which admits a $\poly(n)$-time $\BPPPATH$ algorithm.

	Suppose there is a $\poly(n)$-time $\BPPPATH$ algorithm for a function $g$. Let $x$ be the input. Amplifying the probability gap a bit, we have two polynomials $a(x)$ and $r(x)$ (representing the number of accepting paths and rejecting paths), such that: 
	
	\begin{itemize}
		\item When $g(x)=1$, $a(x) > 3\cdot r(x)$ and $a(x) \ge 1$.
		\item When $g(x)=0$, $a(x) < r(x)/3$ and $r(x) \ge 1$.
		\item We can write $a(x)$ as $a(x) := \sum_{i=1}^{m} a_i \cdot C_i(x)$, such that each $C_i$ is a $\poly(n)$-term, each $a_i$ is non-negative, and for all input $x$, $a(x) \le \exp(\poly(n))$. The same goes for $r(x)$.
	\end{itemize}
	
	The first two claims are straightforward,  and the last claim is due to the fact that one can create at most $\exp(\poly(n))$ possible computation paths in $\poly(n)$ time.\footnote{This is why we need to state $\poly(n)$-time instead of $\poly(n)$-query in Theorem~\ref{theo:hard-adapt}.}
	
	We say a pair of polynomials $a(x)$ and $r(x)$ computes a function $g$ if it satisfies the above three conditions (note it may not present any $\BPPPATH$ algorithms). Then we are going to prove there cannot be such a pair of polynomials for $\Ada_{f,n}$, which refutes the possibility of a $\poly(n)$-time $\BPPPATH$ algorithm as well.
	
	We are going to show that there must be an $x$ such that $a(x) \ge 2^{2^n}$ for a pair of polynomials $a(x)$ and $r(x)$ computing $\Ada_{f,n}$, which contradicts the third condition.
	 
	For each integer $d$, we will inductively construct two distributions  $\distr_1^d$ and $\distr_0^d$ supported on $1$-inputs and $0$-inputs to $\Ada_{f,d}$ respectively, such that $a(\distr_1^d)/r(\distr_0^d) \ge 2^{2^d}$ for any pair of polynomials $a(x)$ and $r(x)$ computing $\Ada_{f,d}$.
	
	The base case $d=0$ is very simple. $\Ada_{f,0}$ is just the $f$ itself. We let $\distr_1^0 = \distr_1$ and $\distr_0^0=\distr_0$. Since $a(x)$ and $r(x)$ are non-negative linear combination of $\poly(n)$-terms, and $\distr_1$ $o(1)$-almost $k$-wise dominates $\distr_0$ for any $k \le \poly(n)$, we must have $r(\distr_1)/r(\distr_0) \ge 1 - o(1)$. Also,  $a(\distr_1)\ge 3 \cdot r(\distr_1)$ as $\distr_1$ is supported on $1$-inputs to $f$. Putting these facts together, we have $a(\distr_1) \ge 2 \cdot r(\distr_0)$, which means $a(\distr_1^0)/r(\distr_0^0) \ge 2=2^{2^0}$.
	
	For $d>0$, suppose that we have already constructed distributions $\distr_{d-1}^0$ and $\distr_{d-1}^1$ on inputs of $\Ada_{f,d-1}$, we are going to construct $\distr_d^0$ and $\distr_d^1$ based on them.
	
	We first decompose the input to $\Ada_{f,d}$ as a triple $(w,x,y) \in D \times D_{d-1} \times D_{d-1}$ as in its definition, in which $D$ denotes the domain of $f$, and $D_{d-1}$ denotes the domain of $\Ada_{f,d-1}$.
	
	For a pair of polynomials $a(w,x,y)$ and $r(w,x,y)$ computing $\Ada_{f,d}$, consider the following two polynomials on $x$: 
	$$
	a_L(x) := a(\distr_0,x,\distr_0^{d-1}) =\Ex_{w \sim \distr_0,y \sim \distr_0^{d-1}} [a(w,x,y)],
	$$ and
	$$
	r_L(x) := r(\distr_0,x,\distr_0^{d-1}) = \Ex_{w \sim \distr_0,y \sim \distr_0^{d-1}} [r(w,x,y)].
	$$
	
	Note that $\distr_0$ is supported on $0$-inputs, which means for any fixed $W \in \support(\distr_0)$ and any $Y \in \support(\distr_0^{d-1})$, by the definition of $\Ada_{f,d}$, the polynomial pair $a(W,x,Y)$ and $r(W,x,Y)$ must compute $\Ada_{f,d-1}$. It is not hard to verify by linearity, that their expectations $a_L(x)$ and $r_L(x)$ also computes $\Ada_{f,d-1}$ (Recall that a pair of polynomials computes a function $g$ if it satisfies the three conditions). 
	
	Therefore, plugging in $\distr_{d-1}^0$ and $\distr_{d-1}^{1}$, we have $a_L(\distr_{d-1}^1) \ge 2^{2^{d-1}} \cdot r_L(\distr_{d-1}^0)$, which means $a(\distr_0,\distr_{d-1}^1,\distr_{d-1}^0) \ge 2^{2^{d-1}} \cdot r(\distr_0,\distr_{d-1}^0,\distr_{d-1}^0)$.
	
	Then, for each fixed $X,Y$, the polynomial $a_M(w) := a(w,X,Y)$ is a non-negative linear combination of $\poly(n)$-terms, since $\distr_1$ $o(1)$-almost $k$-wise dominates $\distr_0$ for any $k \le \poly(n)$, we have $a_M(\distr_1)/a_M(\distr_0) \ge 1-o(1)$. Hence by linearity, $a(\distr_1,\distr^1_{d-1},\distr^0_{d-1}) \ge (1-o(1)) \cdot a(\distr_0,\distr^1_{d-1},\distr^0_{d-1})$.
	
	Now, notice that $\distr_1$ is supported on $1$-inputs to $f$, and $\distr^0_{d-1}$ is supported on $0$-inputs to $\Ada_{f,d-1}$, so $(\distr_1,\distr^1_{d-1},\distr^0_{d-1})$ is supported on $0$-inputs, therefore 
	 $$
	 r(\distr_1,\distr^1_{d-1},\distr^0_{d-1}) \ge 3\cdot a(\distr_1,\distr^1_{d-1},\distr^0_{d-1}) \ge a(\distr_0,\distr^1_{d-1},\distr^0_{d-1})
	 \ge 2^{2^{d-1}} \cdot a(\distr_0,\distr^0_{d-1},\distr^0_{d-1}).
	 $$ 
	 
	 Finally, consider the polynomials on $y$ defined by
	  
	 $$
	 a_{R}(y) := a(\distr_1,\distr^1_{d-1},y)
	 \text{ and } 
	 r_{R}(y) := r(\distr_1,\distr^1_{d-1},y).
	 $$
	 By the same augment as above, they are also a pair of polynomials which computes $\Ada_{f,d-1}$, so plugging in $\distr_{d-1}^0$ and $\distr_{d-1}^1$ again, we have $a_{R}(\distr_{d-1}^1) \ge 2^{2^{d-1}} \cdot r_{R}(\distr_{d-1}^0)$, which means 
	 
	 $$
	 a(\distr_1,\distr^1_{d-1},\distr^1_{d-1}) \ge 2^{2^{d-1}} \cdot r(\distr_1,\distr^1_{d-1},\distr^0_{d-1}) \ge  2^{2^d} \cdot r(\distr_0,\distr^0_{d-1},\distr^0_{d-1}).
	 $$
	 
	 So we can just take $\distr^1_d = (\distr_1,\distr^1_{d-1},\distr^1_{d-1})$ and $\distr^0_d = (\distr_0,\distr^0_{d-1},\distr^0_{d-1})$. It is not hard to see that these distributions are supported on $1$-inputs and $0$-inputs to $\Ada_{f,d}$ respectively.
	
	 Then for a pair of polynomials $a(x)$ and $r(s)$ computing $\Ada_{f,n}$, we have $a(\distr^1_n) \ge 2^{2^n} \cdot r(\distr^0_n) \ge 2^{2^n}$, which means there exists an $x$  such that $a(x) \ge 2^{2^n}$, and this completes the proof.

\end{proof}

\subsection{$\PTIME^{\SZK^\Oracle} \not\subset \BPPPATH^{\Oracle}$}

\newcommand{\fada}{f_{\Ada}}

Let $\fada := \Ada_{\Collision_n,n}$. There is a simple $\PTIME^{\SZK}$ algorithm for $\fada$:  invoke the $\SZK$ oracle $n$ times to decide go to the left child or the right child, and invoke it once again to output the answer to the input on the reached leaf.

Using Theorem~\ref{theo:hard-adapt}, we immediately know $\fada$ is hard for $\poly(n)$ time $\BPPPATH$ algorithms.

\begin{lemma}
	There are no $\poly(n)$-time $\BPPPATH$ algorithms for $\fada$.
\end{lemma}

\begin{proof}
	Note that $\Dtwoone^n$ and $\Dperm^n$ are supported on $1$-inputs and $0$-inputs to $\Collision_n$ respectively, and $\Dtwoone^n$ $o(1)$-almost $k$-wise dominates $\Dperm^n$ for any $k \le \poly(n)$ by Lemma~\ref{lm:col-distr}. Then the lemma directly follows from Theorem~\ref{theo:hard-adapt}.
\end{proof}

Now the following corollary follows directly by a standard diagonalization argument.

\begin{cor}
	There exists an oracle $\Oracle$ such that $\PTIME^{\SZK^\Oracle} \not\subset \BPPPATH^\Oracle$.
\end{cor}

	\section{Oracle Separations from \SZK\ and \QSZK}


In this section we give a simple but powerful method to construct problems which are hard for $\SZK$ or $\QSZK$. The construction is inspired by the cheat sheet functions in \cite{aaronson2015separations}.\footnote{In fact, it is a simpler version of the original construction in \cite{aaronson2015separations}, as there is no need to certify the input domain.}

\newcommand{\xpart}{instance}
\newcommand{\ypart}{check-bit}

\begin{defi}[The check-bit construction]
Let $f : D \to [R]$ be a function such that $D \subset \{0,1\}^M$. We define its check-bit version $\chk{f}$ as follows:

$\chk{f}$ is a function from $D \times \{0,1\}^R \to \{0,1\}$. We call the first part of its input as the \xpart\ part and the second part as the \ypart\ part. For $(x,y) \in D \times \{0,1\}^R$, we define $\chk{f}(x,y) = y_{f(x)}$, that is, the $f(x)^{th}$ bit in the \ypart\ part.

When we only have a Boolean function $f:D \to \{0,1\}$, we can take $c$ copies of it to get a function $f^{\otimes c}: D^{c} \to [2^c]$ (we can fix a bijection between $\{0,1\}^c$ and $[2^c]$). And apply the above construction to get function $\chk{f^{\otimes c}}$.
\end{defi}

Given a function $f$ which has a large image, and needs a lot of queries to evaluate for a randomized algorithm or a quantum one. Then its check-bit version, $\chk{f}$, should be hard for a $\SZK$ protocol or a $\QSZK$ protocol as well. Since intuitively, if the prover want to convince the verifier that the $\ell^{th}$ bit is $1$ for $\ell = f(x)$, she must send some information about $\ell$, but $\ell$ is very hard for the verifier to obtain herself, as $f$ is hard for randomized or quantum algorithms. So it would contradict the zero-knowledge requirement.

\subsection{$\BQP^{\Oracle} \not\subset \SZK^{\Oracle}$ based on \PSimon's problem}

\newcommand{\fsimon}{f_{\Simon}}

By interpreting $\Simon_n$'s output as an integer in $[2^n] = [N]$, we can construct its check-bit version $\fsimon := \chk{\Simon}$.

There is a trivial $\poly(n)$-query quantum algorithm for $\fsimon$: compute the $\Simon_n$ function with input given in the \xpart\ part, then output the corresponding bit in the \ypart\ part.

We are going to show that there are no efficient $\SZK$ protocols for $\fsimon$.

\begin{lemma}
	There are no $\SZK$ protocols for $\fsimon$ in which the verifier makes only $\poly(n)$ queries to the input.
\end{lemma}

\begin{proof}
	For the contradiction, suppose there is such a $\SZK$ protocol for $\fsimon$, in which the verifier makes only $\poly(n)$ queries to the input. Without loss of generality, we can assume the verifier always makes exactly $T \le \poly(n)$ queries to the input.
	
	Now, based on that protocol, we are going to construct a randomized algorithm for $\Simon_n$ with only $\poly(n)$ queries but $1/\poly(n)$ success probability, which clearly contradicts Lemma~\ref{lm:hard-simon} as $N =2^n$.
	
	Let the input to $\fsimon$ be $z=(x,y)$. By the completeness result by Sahai and Vadhan  \cite{sahai2003complete}, such a protocol implies that we have two distributions $\mu_1(z)$ and $\mu_2(z)$, such that one can generate a sample from them using only $\poly(n)$ queries to the input, and $\|\mu_1(z) - \mu_2(z) \| \ge 1 - 2^{-n}$ when $\fsimon(z)=1$, while $\|\mu_1(z) - \mu_2(z)\| \le 2^{-n}$ when $\fsimon(z)=0$.
	
	Let $x$ be a valid input to $\Simon_n$, $y$ be the all-zero string in $\{0,1\}^N$, and $y'$ be the string obtained by changing the $\Simon_n(x)^{th}$ bit in $y$ to $1$. Then by the definition of $\fsimon$, we can see that $\|\mu_1(x,y)-\mu_2(x,y)\| \le 2^{-n}$ and $\|\mu_1(x,y') - \mu_2(x,y')\| \ge 1-2^{-n}$. By triangle inequality, we can see that either $\|\mu_1(x,y) - \mu_1(x,y')\| \ge 1/3$ or $\|\mu_2(x,y) - \mu_2(x,y')\| \ge 1/3$.
	
	We can now describe our algorithm, we first guess a random index $i \in [2]$, so with probability $1/2$, we have $\|\mu_i(x,y) - \mu_i(x,y')\| \ge 1/3$. But since $y$ and $y'$ only differs at the position $\ell = \Simon_n(x)$, when drawing sample from $\mu_i(x,y)$, it must query the $\ell^{th}$ bit of $y$ with probability at least $1/3$, for otherwise $\|\mu_i(x,y) - \mu_i(x,y')\|$ would be smaller than $1/3$. So we simply draw a sample from $\mu_i(x,y)$, and output randomly an index in the \ypart\ part which the sampling algorithm $\mu_i$ has queried. As discussed above, this algorithm computes $\Simon_n$ with probability at least $1/\poly(n)$, and this completes the proof.
\end{proof}

\subsection{$\PTIME^{\SZK^{\Oracle}} \not\subset \QSZK^{\Oracle}$}

\newcommand{\fcol}{f_{\Collision}}

Let $c=10n$. We are going to use the following function: $\fcol := \chk{\Collision^{\otimes c}}$, the check-bit version of  $\Collision_n^{\otimes c}$.

There is a simple $\PTIME^{\SZK}$ algorithm for it: Given input $z = (x,y)$, invoke the $\SZK$ oracle for $c$ times to calculate $\ell = \Collision^{\otimes c}(x)$, then output the $\ell^{th}$ bit of $y$.

We are going to show that there cannot be any efficient $\QSZK$ protocols for $\fcol$. The following proof is similar to the proof of Theorem~12 in \cite{aaronson2015separations}.

We will need the following strong direct product theorem due to Lee and Roland \cite{lee2013strong}.
\begin{theo}
	\label{theo:sdpt}
	Let $f$ be a (partial) function with ${Q}_{1/4}\left(f\right)\geq T'$. \
	Then any $T'$-query quantum algorithm evaluates $c$
	copies of $f$ with success probability at most $O(\left(3/4\right)^{c/2})$.
\end{theo}

\begin{lemma}
	There are no $\QSZK$ protocols for $\fcol$ in which the verifier only makes $\poly(n)$ queries to the input.
\end{lemma}

\begin{proof}
	
	By Theorem~\ref{theo:sdpt}, and $Q_{1/4}(\Collision_n) = \Omega(N^{1/3}) = \Omega(2^{n/3})$, we can see any quantum algorithms with $\poly(n)$ queries can solve $\Collision^{\otimes c}$ with success probability at most $O(\left(3/4\right)^{c/2}) = O(2^{-n})$ (recall that $c = 10n$). 
	
	For the contradiction, suppose there is a $\QSZK$ protocol for $\fcol$ such that the verifier makes only $\poly(n)$ queries to the input. Then we are going to show there is a quantum algorithm with $\poly(n)$ queries computing $\Collision^{\otimes c}$ correctly with probability at least $1/\poly(n)$, which contradicts Theorem~\ref{theo:sdpt}.
	
	Let the input be $z=(x,y)$, in which $x$ is an input to $\Collision^{\otimes c}$ and $y \in \{0,1\}^{2^c}$.
	As shown by Watrous \cite{watrous2002limits}, such a $\QSZK$ protocol implies that there are two mixed quantum states $\xi_1(z)$ and $\xi_2(z)$, which can both be prepared using $\poly(n)$ queries to the input, such that $\| \xi_1(z) - \xi_2(z) \|_{tr} \ge 1-2^{-n}$ when $\fcol(z) = 1$, and $\| \xi_1(z) - \xi_2(z) \|_{tr} \le 2^{-n}$ when $\fcol(z) = 0$.
	
	Now, let $x$ be a valid input to $\Collision^{\otimes c}$, then consider running $\xi_1$ and $\xi_2$ on input $z=(x,y)$, such that $y=0^{2^c}$. Clearly, by definition, we have $\fcol(z)=0$ hence $\|\xi_1(z) - \xi_2(z)\|_{tr} \le 2^{-n}$ in that case. Let $\ell = \Collision^{\otimes c}(x)$, then if we change the $\ell^{th}$ bit of $y$ to $1$, we immediately get an input $z'$ such that $\fcol(z') = 1$, so $\|\xi_1(z') - \xi_2(z')\|_{tr} \ge 1- 2^{-n}$. By triangle inequality, we have either $\|\xi_1(z) - \xi_1(z')\|_{tr} \ge 1/3$ or $\|\xi_2(z) - \xi_2(z')\|_{tr} \ge 1/3$.
	
	Now we describe our algorithm for computing $\Collision^{\otimes c}$ with a non-negligible probability. Without loss of generality, we can assume both $\xi_1$ and $\xi_2$ require exactly $T \le \poly(n)$ queries to prepare. 
	
	Given an input $x$ to $\Collision^{\otimes c}$, let $\ell=\Collision^{\otimes c}(x)$. For all $i \in [2]$, $w \in [2^c]$ and $t \in [T]$, we define the {\em query magnitude} $m_{i,w,t}$, to be the probability that the preparation algorithm for $\xi_i$ would be found querying the $w^{th}$ bit in the \ypart\ part of the input, were we to measure in the standard basis before the $t^{th}$ query, when it is applied to input $z=(x,0^{2^c})$.
	
	We first guess a random index $i \in [2]$. Then as discussed above, with probability $1/2$, we have $\|\xi_{i}(z) - \xi_{i}(z')\|_{tr} \ge 1/3$, in which  $z'$ is obtained by changing the $\ell^{th}$ bit to $1$ in the \ypart\ part of $z$. Using these facts and the hybrid argument in \cite{bennett1997strengths}, it follows that
	
	$$
	\sum_{t=1}^{T} \sqrt{m_{i,\ell,t}} \ge \Omega(1).
	$$ 
	
	Then by Cauchy-Schwarz inequality, we have
	$$
	\sum_{t=1}^{T} m_{i,\ell,t} \ge \Omega\left(\frac{1}{T}\right).
	$$
	
	This means that, if we pick a random $i \in[2]$ and $t \in [T]$, run $\xi_i$ until the $t^{th}$ query on input $z=(x,0^{2^c})$, and then measure in the standard basis, we will observe $\ell = \Collision^{\otimes c}(x)$ with probability at least $\Omega(1/T^2)$. Then we get an algorithm computing $\Collision^{\otimes c}$ with $\poly(n)$ queries and at least $1/\poly(n)$ probability. This completes the proof.
	
\end{proof}



\section{Acknowledgment}

We would like thank Scott Aaronson for several helpful discussions during this work and detailed comments on an early draft of this paper.
	
	\bibliographystyle{alpha}
	\bibliography{team} 
	
\end{document}